\documentclass[3p,final]{elsarticle}
\usepackage[utf8]{inputenc}
\usepackage{amssymb}
\usepackage{amsmath}
 \usepackage{amsthm}
\usepackage{hyperref}
\usepackage{graphicx}
\usepackage{natbib}

\journal{arXiv}

\newtheorem{observation}{Observation}
\newtheorem{lemma}{Lemma}
\newtheorem{definition}{Definition}
\newtheorem{theorem}{Theorem}
\newtheorem{remark}{Remark}

\begin{document}

\begin{frontmatter}
\title{A Linear Time Algorithm for Computing the Eternal Vertex Cover Number of Cactus Graphs}
\author{Jasine Babu}
\ead{jasine@iitpkd.ac.in}
\author{Veena Prabhakaran}
\ead{veenaprabhakaran7@gmail.com}
\author{Arko Sharma}
\ead{arkosharma4@gmail.com}
\address{Indian Institute of Technology Palakkad, Kerala, India 678557}

\begin{abstract}
The eternal vertex cover problem is a dynamic variant of the classical vertex
cover problem. It is NP-hard to compute the eternal vertex cover number of graphs and known algorithmic results for the problem are very few. 
This paper presents a linear time recursive algorithm for computing the eternal vertex cover number of cactus graphs. Unlike other graph classes for which polynomial time algorithms for eternal vertex cover number are based on efficient computability of a known lower bound directly derived from minimum vertex cover, we show that it is a certain substructure property that helps the efficient computation of eternal vertex cover number of cactus graphs. An extension of the result  to graphs in which each block is an edge, a cycle or a biconnected chordal graph is also presented.  
\end{abstract}
\end{frontmatter}
\newcommand{\evc}{\operatorname{evc}}
\newcommand{\mvc}{\operatorname{mvc}}
\newcommand{\cvc}{\operatorname{cvc}}
\newcommand{\vb}{\operatorname{VB_G}}
\section{Introduction}
Eternal vertex cover problem can be described in terms of a two player game played on any given graph $G$ between an attacker who attacks the edges of $G$ and a defender who controls the placement of a fixed number of guards (say $k$) on a subset of vertices of $G$ to defend the attack \cite{Klostermeyer2009}. In the beginning of the game, the defender chooses a placement of guards on a subset of vertices of $G$, defining an \textit{initial configuration}. Subsequently, in each round of the game, the attacker attacks an edge of her choice. In response, the defender has to reconfigure the guards. The guards move in parallel. Each guard may be retained in its position or moved to a neighboring vertex, ensuring that at least one guard from an endpoint of the attacked edge moves to its other endpoint.  
If the defender is able to do this, then the attack is said to be successfully defended and the game proceeds to the next round with the attacker choosing the next edge to attack\footnote{There are two popular versions of the problem: one as described here and the second with an additional constraint that not more than one guard can be on any vertex in any configuration. The number of guards required to protect the graph may differ between these two versions. However, the results derived in this paper and those cited here hold in both the versions.}. Otherwise, the attacker wins. If the defender is able to keep defending against any infinite sequence of attacks on $G$ with $k$ guards,
then we say that there is a defense strategy on $G$ with $k$ guards. In this case, the positions of guards in any round of the game defines a \textit{configuration}, referred to as an \textit{eternal vertex cover} of $G$ of size $k$. The set of all configurations encountered in the game defines an \textit{eternal vertex cover class} of $G$. The \textit{eternal vertex cover number} of a graph $G$, denoted by $\evc(G)$ is the minimum integer $k$ such that there is a defense strategy on $G$ using $k$ guards. It is clear from the description of the game that in any configuration, at least one of the endpoints of each edge must have a guard for defending attacks successfully. Hence, it is clear that each eternal vertex cover of $G$ is a vertex cover of $G$ and $\evc(G) \ge \mvc(G)$, the vertex cover number of $G$. 

The decision version of the eternal vertex cover problem takes a graph $G$ and an integer $k$ as inputs and asks whether $G$ has an eternal vertex cover of size at most $k$. Fomin et al.~\cite{Fomin2010} showed that this problem is NP-hard and is in PSPACE. The problem was recently shown to be NP-Complete for internally triangulated planar graphs \cite{evc_lb_Arxiv}. It is not yet known whether the problem is in NP for bipartite graphs~\cite{Fomin2010}. Similarly, it is an open question whether computing the eternal vertex cover number of bounded treewidth graphs is NP-Hard or not~\cite{Fomin2010}. The problem remains open even for graphs of treewidth two. 

Dynamic variants of other classical graph parameters like dominating set~\cite{GOLDWASSER2008,KLOSTERMEYER2012,rinemberg2019,goddard2005} and independent set~\cite{Hartnell2014,Caro2016EternalIS}  are also well studied in literature.
Rinemberg et al.~\cite{rinemberg2019} recently showed that eternal domination number of interval graphs can be computed in polynomial time.
The relationship between eternal vertex cover number and other dynamic graph parameters was explored by Klostermeyer and Mynhardt~\cite{Klostermeyer2011} and Anderson et al.~\cite{Anderson2014}. It is known that $\evc(G) \le \cvc(G)+1$, where $\cvc(G)$ is the minimum cardinality of a connected vertex cover of $G$ \cite{Klostermeyer2009, Fomin2010}. Fomin et al.~\cite{Fomin2010} showed that $\evc(G)$ is at most twice the size of a maximum matching in $G$. 
It is also known that for any connected graph $G$, $\evc(G) \ge \mvc_X(G)$, where $\mvc_X(G)$ is the minimum cardinality of a vertex cover of $G$ that contains all cut vertices of $G$~\cite{evc_lb_Arxiv}. 

Precise formulae or algorithms for computing eternal vertex cover number are known for only a few classes of graphs. The initial work by Klostermeyer and Mynhardt gave explicit formulae for the eternal vertex cover number of trees, cycles and grids~\cite{Klostermeyer2009}. 
The eternal vertex cover number of a tree $T$ is known to be $|V(T)|-|L(T)|+1$, where $L(T)$ is the number of leaves of $T$. The eternal vertex cover number of a cycle on $n$ vertices is $\lceil\frac{n}{2}\rceil$, which is equal to its vertex cover number. A polynomial time algorithm for computing eternal vertex cover number of chordal graphs has been obtained recently~\cite{evc_lb_Arxiv}. Parameterized algorithms for the problem are discussed by Fomin et al.~\cite{Fomin2010}. Araki et al.~\cite{Hisashi_Gen_Trees} discusses a polynomial time algorithm to compute eternal vertex cover number of a simple generalization of trees. 

All the graph classes for which polynomial time algorithms for eternal vertex cover are known so far exhibit a common property. Recall that 
$\mvc_X(G)$, which is the minimum cardinality of a vertex cover of $G$ that contains all cut vertices of $G$ is a lower bound to $\evc(G)$.
For any graph $G$ that belong to the classes given above, this lower bound  is polynomial time computable and $\evc(G) \in \{\mvc_X(G), \mvc_X(G)+1\}$. However, this is not true for cactus graphs, as shown in Section~\ref{sec:example}. Hence, we formulate a new lower bounding technique based on the {\it substructure property} (Definition~\ref{def:substructure}) to handle cactus graphs.
 
A \textit{cactus} is a connected graph in which any two simple cycles have at most one vertex in common. A cactus is a generalization of cycles and trees in the following sense: if we contract every cycle in a cactus to a single vertex, the resultant graph is a tree. Cactus graphs have treewidth at most two. Many important NP-hard combinatorial problems are polynomial time solvable on cactus graphs~\cite{HARE1987437}. Since trees and cycles are examples of cactus graphs, we know that there exist cactus graphs $G$ for which $\mvc_X(G)$ lower bound is close to $\evc(G)$. However, there are cactus graphs $G$ for which the eternal vertex cover number is more than $1.5$ times $\mvc_X(G)$  (Section~\ref{sec:example}).  Thus, the reason for efficient computability of eternal vertex cover number for this graph class is not just related to the lower bound $\mvc_X(G)$. 

The lower bound argument for eternal vertex cover number of trees was obtained using a recursive procedure~\cite{Klostermeyer2009}. 
We show that by using a more sophisticated recursive method based on the substructure property of cactus graphs, we can device a new lower bounding strategy applicable for cactus graphs. 
From this, we derive a formula for computing the eternal vertex cover number of a cactus $G$, in terms of eternal vertex cover number and some related parameters of certain edge disjoint subgraphs of $G$. The algorithm for computing eternal vertex cover number and the other required parameters is a recursive procedure based on this formula and it runs in time linear in the size of the input graph. Finally, we present an extension of the result  to graphs in which each block is an edge, a cycle or a biconnected chordal graph. It is shown that for graphs of this class, eternal vertex cover number can be computed in quadratic time. In particular, this method implies a quadratic time algorithm for the computation of eternal vertex cover number of chordal graphs, which is different from the algorithm given in~\cite{evc_lb_Arxiv}.
\section{Some Basic Observations}\label{sec:basics}
\begin{definition}
 Let $G$ be a graph and $S \subseteq V(G)$. The minimum cardinality of a vertex cover of $G$ that contains all vertices of $S$ is denoted by $\mvc_S(G)$. 
 The minimum integer $k$ such that there is a defense strategy on $G$ using $k$ guards with all vertices of $S$ being occupied in each configuration is denoted by 
 $\evc_S(G)$. 
\end{definition}
When $S=\{v\}$, we use $\mvc_v(G)$ and $\evc_v(G)$ respectively instead of $\mvc_S(G)$ and $\evc_S(G)$. From the above definition, it is clear that $\evc(G) \le \evc_S(G)$.
\begin{definition}
 If $G$ is a graph and $x \in V(G)$, we use $G_x^+$ to denote the graph obtained by adding an additional vertex 
which is made adjacent only to $x$. 
\end{definition}
\begin{observation}\label{obs:evc_extension_ub}
$\evc(G_x^+) \le \evc(G)+1$.
\end{observation}
\begin{proof} 
  Suppose $\evc(G)=k$. Let $v$ be the vertex in $V(G_x^+) \setminus V(G)$. 
   There is a defense strategy on $G_x^+$ using $k+1$ guards, by using $k$ guards to move exactly the same way as they would have been moved for defending attacks on $G$ and 
   an additional $k+1^{th}$ guard to protect against attacks on the edge $xv$. If the edge  $xv$ is attacked when $x$ is not occupied, then
  the additional guard will move from $v$ to $x$. In the very next round, this guard moves back to $v$, irrespective of the edge attacked. 
  It is easy to check that using this strategy, any attack on $G_x^+$ can be defended.
 \end{proof}

\begin{definition}[x-components and x-extensions]\label{def:x-component}
Let $x$ be a cut vertex in a connected graph $G$ and $H$ be a component of $G \setminus x$. Let $G_0$ be the induced subgraph of $G$
on the vertex set $V (H) \cup \{x\}$. Then, $G_0$ is called an $x$-component of $G$ and $G$ is
called an $x$-extension of $G_0$.
%
%
\end{definition}

 \begin{definition}[Substructure property]\label{def:substructure}
 Let $x$ be an arbitrary non-cut vertex of a graph $G$. If the following is true for any arbitrary $x$-extension $G'$ of $G$, then $G$ satisfies substructure property:
 \begin{itemize}
  \item  if  $\evc(G_x^+) \le \evc_x(G)$, then in every eternal vertex cover $C'$ of $G'$, the number of guards on $V(G)$ is at least $\evc_x(G)-1$ and
  \item  if $\evc(G_x^+) > \evc_x(G)$, then in every eternal vertex cover $C'$ of $G'$, the number of guards on $V(G)$ is at least $\evc_x(G)$.
 \end{itemize}
\end{definition}

\begin{observation}\label{obs:evcx_lb}
 If a graph $G$ satisfies substructure property and $x$ is a non-cut vertex of $G$, then $\evc_x(G) \le \evc(G_x^+) \le \evc_x(G)+1$. 
\end{observation}
\begin{proof}
 Suppose $G$ satisfies substructure property and $x$ is a non-cut vertex of $G$. Let $v$ be the vertex in $V(G_x^+)\setminus V(G)$.  In any eternal vertex cover class of $G_x^+$, there must be a configuration in which there is a guard on $v$. 
 By substructure property, in this configuration, the number of guards in $V(G)$ must be at least $\evc_x(G)-1$.  Hence, the number of guards required on 
 $V(G_x^+)$ is at least $\evc_x(G)$.  Consequently, $\evc(G_x^+) \ge \evc_x(G)$. Further, by Observation~\ref{obs:evc_extension_ub}, it follows that $\evc(G_x^+) \le \evc_x(G)+1$. 
\end{proof}

\begin{definition}
 Let $G$ be a graph and $x$ be a vertex of $G$.  
 \begin{enumerate}
  \item $G$ is Type~1 with respect to $x$ if $\evc(G_x^+)=\evc_x(G)$ and
  \item $G$ is Type~2 with respect to $x$ if $\evc(G_x^+)>\evc_x(G)$. 
 \end{enumerate}
\end{definition}
\begin{remark}
 From the above definition, a tree is Type~1 with respect to any of its vertices.
 An even cycle is Type~1 with respect to any of its vertices. An odd cycle is Type~2 with respect to any of its vertices. A complete graph on $3$ or more vertices is Type~2 with respect to any of its vertices. 
\end{remark}
\begin{remark}\label{rmk:type_noncut}
 By Observation~\ref{obs:evcx_lb}, if $G$ satisfies substructure property and $x$ is a non-cut vertex of $G$, then $G$ is either Type~1 or Type~2 with respect to $x$. 
 \end{remark}
 \begin{definition}[\cite{evc_lb_Arxiv}]
 Let $x$ be a cut vertex of a connected graph $G$.  
 The set of $x$-components of $G$ will be denoted as $\mathcal{C}_x(G)$. For $i \in \{1, 2\}$, we define $T_i(G, x)$ to be the set of all $x$-components of $G$ that are Type~$i$ with respect to $x$.  
\end{definition}
\begin{lemma}\label{lem:cut-vertex}
 Let $G$ be a connected graph and $x$ be a cut-vertex of $G$ such that each $x$-component of $G$ satisfies the substructure property.
 \begin{enumerate}
  \item If all $x$-components of $G$ are Type~2, then $\evc(G)=\evc_x(G)=1+\sum_{G_i \in \mathcal{C}_x(G)}{(\evc_x(G_i)-1)}$.
  \item Otherwise, $\evc(G)=\evc_x(G)=2+\sum_{G_i \in T_1(G, x)}{(\evc_x(G_i)-2)}+
  \sum_{G_i \in T_2(G, x)}{(\evc_x(G_i)-1)}$.
 \end{enumerate}
\end{lemma}
\begin{proof}
 The first part is easy to see from the definition of the substructure property. 
 
 To prove the second part, suppose $H$ is an $x$-component of $G$ which is Type~1 with respect to $x$. Let $k=1+\sum_{G_i \in T_1(G, x)}{(\evc_x(G_i)-2)}+
  \sum_{G_i \in T_2(G, x)}{(\evc_x(G_i)-1)}$. By a direct application of the substructure property,
 it follows that $\evc(G)$ cannot be less than $k$. 
  Now, for contradiction, suppose $\evc(G)=k$. Let $\mathcal{C}$ be a minimum eternal vertex cover class $G$. 
  Since there are only $k$ guards on $V(G)$ and each $x$-component of $G$ satisfies the substructure property, 
  it can only be the case that in every configuration of $\mathcal{C}$ the number of guards on an $x$-component $G_i$ is exactly 
  $\evc_x(G_i)-1$ if $G_i$ is Type~1 with respect to $x$ and it is exactly $\evc_x(G_i)$ if $G_i$ is Type~2 with respect to $x$.
  Further, $x$ must also be occupied in every configuration. 
  However, to defend repeated attacks on edges of $H$ which is a Type~1 $x$-component with respect to $x$, 
  $\evc_x(H)$ guards have to be on $V(H)$ in some configuration, if $x$ is to be always occupied. 
  Then, the total number of guards on $V(G)$ must be more than $k$, a contradiction. Hence, $\evc(G) \ge k+1$.  
  
  It is not difficult to show that $\evc_x(G) \le k+1$. 
  In every configuration, we will maintain the following invariants. 
  \begin{itemize}
   \item  A guard will be kept on $x$.
   \item  There will be exactly $\evc_x(G_i)$ guards on each $x$-component of $G$ that is Type~2 with respect to $x$ such that 
  the guards on $V(G_i)$ form an eternal vertex cover of $G_i$.
  \item  In one of the Type~1 $x$-components $G_i$ of $G$ (if exists), exactly $\evc_x(G_i)$ guards will be kept
  and in every other Type~1 component $G_j$ of $G$, $\evc_x(G_j)-1$ guards will be kept. This will be done in such a way that the  
  induced configuration on $V(G_i)$, for each $x$-component $G_i$ forms an induced configuration of an eternal vertex cover of $(G_i)_x^+$. 
  \end{itemize}
   To defend an attack on an edge in a Type~2 component $G_i$ maintaining the invariants, it will be enough to move guards only on $V(G_i)$.
  To defend an attack on an edge in a Type~1 component, it will be enough to move guards in at most two $x$-components. If the edge attacked is in a Type~1 $x$-component $G_i$ 
  with $\evc_x(G_i)-1$ guards on $V(G_i)$, then there is another Type~1 $x$-component $G_j$ with $\evc_x(G_j)$ guards on $G_j$. Since $G_i$ and $G_j$ satisfy substructure property,
  it is not difficult to rearrange guards in $V(G_i)$ and $V(G_j)$ such that the number of guards on $V(G_i)$ becomes $\evc_x(G_i)$ and that on $V(G_j)$ becomes $\evc_x(G_j)-1$ and the invariants are
  maintained.
\end{proof}
\section{Computation of the Type of vertices}\label{sec:type-computation}
Lemma~\ref{lem:cut-vertex} indicates the possibility of a recursive method to compute eternal vertex cover number of graphs whose $x$-components satisfy the substructure
property.  However, this makes it necessary to also compute $\evc_x(G)$ and the type of each $x$-component of $G$, with respect to $x$. Since 
a cut vertex of $G$ is not a cut-vertex in its $x$-components, a general method to find the type of a graph with respect to any arbitrary vertex of the graph (including non-cut vertices) is necessary.  
This section addresses this issue systematically.


\begin{observation}[Type with respect to a cut vertex]\label{obs:type_cut_vertex}
  Let $G$ be a connected graph and $x$ be a cut-vertex of $G$ such that each $x$-component of $G$ satisfies the substructure property. 
  Then, $G$ is Type~1 with respect to $x$ if at least one of the $x$-components of $G$ is Type~1 with respect to $x$. Otherwise, $G$ is Type~2 with respect to $x$.
\end{observation}
\begin{proof}
 Note that when a pendent edge $xv$ is added to $x$, that edge is a Type~1 component with respect to $x$. Further, $\evc_x$ of this $x$-component is $2$.
 Using these facts in the expressions given in Lemma~\ref{lem:cut-vertex} immediately yields the observation. 
 \end{proof}
 \begin{observation}
  Let $G$ be a graph that satisfies substructure property and $x$ be any vertex of $G$. If $\evc(G) < \evc_x(G)$, then $G$ is Type~1 with respect to $x$. 
 \end{observation}
\begin{proof}
If $\evc(G) < \evc_x(G)$, then by Observation~\ref{obs:evc_extension_ub}, we have $\evc(G_x^+) \le \evc(G)+1 \le \evc_x(G)$. By Remark~\ref{rmk:type_noncut} and Observation~\ref{obs:type_cut_vertex}, $\evc(G_x^+) \ge \evc_x(G)$. 
Therefore, $\evc_x(G)=\evc(G_x^+)$. From this, the observation follows. 
\end{proof}
The following observation is useful for deciding the type of a graph with respect to a pendent vertex. 
\begin{lemma}[Type with respect to a pendent vertex]\label{lem:pendent_vertex}
 Let $G$ be a graph and $x\in V(G)$.  Let $H=G_x^+$ with $v$ being the vertex in $V(H) \setminus V(G)$. Suppose each $x$-component of $H$ satisfies
 substructure property. Then, $H$ is Type~1 with respect to $v$ if and only if $G$ is Type~1 with respect to $x$.  
\end{lemma}
\begin{proof}
In any eternal vertex cover class of $H$ with $\evc_v(H)$ guards in which $v$ is occupied in every configuration\footnote{In the version of the problem where more than one guard is allowed on a vertex, we still can assume without loss of generality that $v$ has only one guard in any configuration. This is because instead of placing more than one guard on $v$, it is possible to place all but one of those guards on $x$.}, $x$ must also be occupied in every configuration (otherwise, an attack on the edge $vx$ cannot be defended maintaining a guard on $v$).   Hence, the induced configurations on $G$ define an
eternal vertex cover class of $G$ in which $x$ is occupied in every configuration. It follows that $\evc_v(H)\ge \evc_x(G)+1$.  
Moreover, from an eternal vertex cover class of $G$ with $x$ always occupied,
we can get an eternal vertex cover class of $H$ with $v$ always occupied by
placing an additional guard at $v$.  Hence, $\evc_v(H)\leq \evc_x(G)+1$.
Thus, $\evc_v(H)=\evc_x(G)+1$.
Note that, by Observation~\ref{obs:evc_extension_ub}, $\evc(H_v^+) \le \evc(H)+1$.

First, suppose $G$ is Type~1 with respect to $x$. Then, $\evc(H)=\evc_x(G)$ and we get $\evc_v(H)=\evc_x(G) + 1 = \evc(H)+1=\evc(H_v^+)$, which means that $H$ is Type~1 with respect to $v$.

Now, suppose $G$ is Type~2 with respect to $x$. Then, $\evc(H)=\evc_x(G)+1$. 
Further, by Observation~\ref{obs:type_cut_vertex}, if $x$ is a cut vertex in $G$, then all the $x$-components of $G$ are Type~2 with respect to $x$.
Irrespective of whether $x$ is a cut-vertex of $G$ or not, by Lemma~\ref{lem:cut-vertex}, $\evc(H_v^+)=2+\evc_x(G)$. Hence, $\evc(H_v^+)=\evc_v(H)+1$.  
Thus, if $G$ is Type~2 with respect to $x$, then $H$ is Type~2 with respect to $v$.
\end{proof}
Now we give some observations that are useful for deciding the type of a graph with respect to a degree-$2$ vertex.
\begin{lemma}\label{lem:ContractVertex}
 Let $G$ be any graph and suppose $v$ is a degree-$2$ vertex in $G$ such that its neighbors $v_1$, $v_2$ are not adjacent. Let $G'$ be the graph obtained by deleting $v$ and adding an edge between $v_1$ and $v_2$. 
 Then, $\evc_v(G)=\evc(G')+1$.
\end{lemma}
\begin{proof}
 We first give a proof for the version of the problem where at most one guard is allowed on a vertex. 
 
Consider a minimum eternal vertex cover class $\mathcal{C}'$ of $G'$. 
Note that, at least one among $v_1$ and $v_2$  is occupied in each configuration of $\mathcal{C}'$. Let $\mathcal{C}=\{S'\cup \{v\} : S' \in \mathcal{C}' \}$. 
It is easy to see that each $S \in \mathcal{C}$ is a valid vertex cover of $G$.
If $S'_1$, $S'_2$ are configurations in $\mathcal{C}'$ obtainable from each other by a single step of valid movement of guards, it is also easy to verify that the corresponding configurations $S'_1\cup\{v\}$ and $S'_2 \cup \{v\}$ of $G$ are
obtainable from each other by a single step of valid movement of guards. 
From this, it follows easily that $\mathcal{C}$ is an eternal vertex cover class of $G$ with $v$ always occupied. Hence, $\evc_v(G) \le \evc(G')+1$. 

The correspondence between configurations of $G'$ and $G$ works in the reverse direction as well. Let $\mathcal{C}$ be an eternal vertex cover class of $G$  
in which $v$ is permanently occupied. 
Let $\mathcal{C}'=\{S \setminus \{v\} : S \in \mathcal{C} \}$.
Any configuration $S$ in $\mathcal{C}$ must contain at least one of $v_1$ and $v_2$. Therefore, the corresponding configuration $S'= S \setminus \{v\}$ is a vertex cover of $G'$.  Moreover, whenever the guards in $G$ move from a configuration ${S}_1$ to ${S}_2$ in $\mathcal{C}$ via a reconfiguration of guards which involves moving the guard from $v$ to $v_2$ and $v_1$ to $v$, 
we can simulate the behavior in $G'$ by assuming that the guard on $v_1$ is moving to $v_2$ (similarly while moving from $v$ to $v_1$ and $v_2$ to $v$). This corresponds
to a valid movement of guards in $G'$  from configuration $S_1 \setminus\{v\}$ to $S_2 \setminus \{v\}$. From this, it follows easily that $\mathcal{C}'$ is an eternal vertex cover class of $G'$. Thus, $\evc(G') \le \evc_v(G)-1$. 

Now, let us consider the version of the problem where more than one guard is allowed on a vertex. It is easy to modify the first part of the proof and show that $\evc_v(G) \le \evc(G')+1$. In the second part of the proof, if $S$ is a configuration in $\mathcal{C}$ such that for each $x \in V(G)$ there are exactly $t_x$ guards on $x$, then define the corresponding configuration $S'$ of $\mathcal{C}'$ to have $t_{v_1}+t_v-1$ guards on $v_1$ and for each $x \in V(G)\setminus\{v, v_1\}$, $t_x$ guards on $x$. This modification is sufficient to prove $\evc(G') \le \evc_v(G)-1$. 
\end{proof}
 \begin{definition}
 Let $X$ be the set of cut vertices of a graph $G$. If $B$ is a block of $G$, the set of $B$-components of $G$ is defined as 
 $\mathcal{C}_B(G)= \left \{ G_i: G_i \in \mathcal{C}_x(G) \text{ for some $x \in X \cap V(B)$} \text{ and $G_i$ edge disjoint with $B$}\right\}$.
 If $P$ is a path in $G$, then the set of $P$-components of $G$ is defined as\\
 $\mathcal{C}_P(G)=  \{ G_i: G_i \in \mathcal{C}_x(G) \text{ for some $x \in X \cap V(P)$} \text{ and $G_i$ edge disjoint}$ $\text{with $P$}\}$. 
 \end{definition}
 \begin{definition}
 For a block $B$ (respectively, a path $P$) of connected graph $G$, the type of a $B$-component (respectively, $P$-component) is its type with respect to the common vertex it has with $B$ (respectively, $P$). For $i \in \{1, 2\}$, we define $T_i(G, B)$ (respectively, $T_i(G, P)$) to be the set of all $B$-components (respectively, $P$-components) of $G$ that are Type~$i$.    
\end{definition}
If $B$ is a block (respectively, $P$ is a path) of a connected graph $G$, such that all $B$-components (respectively, $P$-components) of $G$ satisfy substructure property, then we can easily obtain a lower bound on the total number of guards on $\bigcup_{G_i \in \mathcal{C}_B(G)} V(G_i)$ (respectively, on $\bigcup_{G_i \in \mathcal{C}_P(G)} V(G_i)$) in any eternal vertex cover of $G$ or its extensions. The notation introduced below is to abstract this lower bound. 
 \begin{definition}
  For a block $B$ of connected graph $G$, we define $$\chi(G, B)=|V(B)\cap X|+\sum_{\substack{G_i \in T_1(G, B)\\ x_i \in V(G_i)\cap V(B)}}{(\evc_{x_i}(G_i)-2)}+
  \sum_{\substack{G_i \in T_2(G, B)\\x_i \in V(G_i)\cap V(B)}}{(\evc_{x_i}(G_i)-1)}$$
  Similarly, for a path $P$ of connected graph $G$, we define $$\chi(G, P)=|V(P)\cap X|+\sum_{\substack{G_i \in T_1(G, P)\\ x_i \in V(G_i)\cap V(B)}}{(\evc_{x_i}(G_i)-2)}+
  \sum_{\substack{G_i \in T_2(G, P)\\ x_i \in V(G_i)\cap V(B)}}{(\evc_{x_i}(G_i)-1)}$$.
 \end{definition}
\begin{remark}
 If $B$ is a block (respectively, $P$ is a path) of a connected graph $G$, such that all $B$-components (respectively, $P$-components) of $G$ satisfy substructure property, then the total number of guards on\\$\bigcup_{G_i \in \mathcal{C}_B(G)} V(G_i)$ (respectively, on $\bigcup_{G_i \in \mathcal{C}_P(G)} V(G_i)$) in any eternal vertex cover of $G$ or its extensions is at least $\chi(G, B)$ (respectively, $\chi(G, P)$).
\end{remark}
\begin{definition}[Vertex bunch of a path]
 Let $P$ be a path in a connected graph $G$. The vertex set $V(P) \cup \bigcup_{G_i \in \mathcal{C}_P(G)} V(G_i)$ is the vertex bunch of 
 $P$ in $G$, denoted by $\vb(P)$. 
\end{definition}
\begin{definition}[Eventful path]\label{def:eventful}
 Let $G$ be a connected graph and $X$ be the set of cut vertices of $G$. A path $P$ in a graph $G$ is an eventful path if
 \begin{itemize}
  \item $P$ is either an induced path in $G$ or a path obtained by removing an edge from an induced cycle in $G$.
  \item the endpoints of $P$ are in $X$ and
  \item any subpath $P'$ of $P$ with both endpoints in $X$ has $|V(P') \setminus X|$ even.   
 \end{itemize}
\end{definition}
\begin{lemma}\label{lemBigPath}
Let $G$ be a connected graph and let $P$ be an eventful path in $G$. 
Let $X$ be the set of cut vertices of $G$. If each $P$-component in $\mathcal{C}_P(G)$ satisfies the substructure property, 
then in any eternal vertex cover configuration of $G$, the total number of guards on $\vb(P)$ is at least $\frac{|V(P) \setminus X|}{2}+\chi(G, P)$.
 Moreover, if $\vb(P) \ne V(G)$ and the number of guards on $\vb(P)$ is exactly equal to the above expression, 
 then at least one of the neighbors of the endpoints of $P$ outside $\vb(P)$ has a guard on it.
\end{lemma}
\begin{proof}
 Consider any subpath $P'$ of $P$ such that both endpoints of $P'$ are in $X$ and none of its intermediate vertices are from $X$. 
 Let $C$ be an eternal vertex cover configuration of $G$. 
 Since $P$ is eventful, $|V(P') \setminus X|$ is even and in any vertex cover of $G$, at least $\frac{|V(P') \setminus X|}{2}$ internal
 vertices of $P'$ must be present. Using this along with the substructure property of $P$-components proves the first part of the lemma.
 
  Now, suppose $\vb(P) \ne V(G)$ and the number of guards  in the configuration $C$ on $\vb(P)$ is exactly equal to the expression given in the lemma. 
  Now, for contradiction, let us assume that none of the neighbors of the endpoints of $P$ outside $\vb(P)$ 
  has a guard in $C$. Consider an attack on an edge $xv$, where $x$ is an endpoint of $P$ and $v$ is a neighbor of $x$ outside $\vb(P)$. 
  To defend this attack, a guard must move from $x$ to $v$. Note that, no guards can move to $\vb(P)$ from outside $\vb(P)$. Hence, 
  while defending the attack, the number of guards on $\vb(P)$ decreases at least by one. But, then the new configuration will violate the first part of the lemma.
  Hence, it must be the case that at least one of the neighbors of the endpoints of $P$ outside $\vb(P)$ has a guard in $C$. 
\end{proof}
\begin{definition}[Maximal uneventful path]
 Let $G$ be a connected graph and let $X$ be the set of cut vertices of $G$. 
 A path $P$ in $G$ is a maximal uneventful path in $G$ if 
 \begin{itemize}
  \item $V(P) \cap X = \emptyset$
  \item $|V(P)|$ is odd and
  \item $P$ is a maximal induced path in $G$ satisfying the above two conditions. 
 \end{itemize}
\end{definition}
The next lemma is applicable to any connected graph $G$ that contains a block $B$ which is a cycle such that all $B$-components satisfy the substructure property. 
Since each block of a cactus is either a cycle or an edge, this lemma will be useful for computing the eternal vertex cover number of cactus graphs. The proof of the lemma makes use of the fact that $B$ can be partitioned into a collection of edge disjoint paths which are either eventful paths or maximal uneventful paths.
\begin{lemma}\label{lem:evc_block}
Let $B$ be a cycle forming a block of a connected graph $G$ and let $X$ be the set of cut vertices of $G$. 
 Suppose each $B$-component $G_i$ of $G$ that belongs to $\mathcal{C}_B(G)$ satisfies the substructure property. If $T_1(G, B)=\emptyset$, then $\evc(G)=\left\lceil{\frac{|V(B)\setminus X|}{2}}\right\rceil+\chi(G, B)$. Otherwise, $\evc(G)=\left\lceil{\frac{|V(B)\setminus X|+1}{2}}\right\rceil+\chi(G, B)$.
\end{lemma}
\begin{proof}
Let $|V(B)|=n$ and $|X \cap V(B)|=k$. For each $B$-component $G_i$ of $G$ that belongs to $\mathcal{C}_B(G)$, let $x_i$ be the vertex that $G_i$ has in common with $B$. 
 Let $\mathcal{C}$ be a minimum eternal vertex cover of $G$. Note that the condition stated in Lemma~\ref{lemBigPath} has to simultaneously 
 hold for all subpaths of the cycle that are eventful in $G$. 
 
  Let $l\ge 0$ be the number of maximal uneventful paths in $B$ and let $P_1, P_2, \ldots, P_l$ be these listed in the cyclic order along $B$.
 To protect the edges within each $P_i$, $V(P_i)$ should contain at least $\mvc(P_i)=\lfloor\frac{|V(P_i)|}{2}\rfloor$ guards. If there are 
 exactly $\lfloor\frac{|V(P_i)|}{2}\rfloor$ guards on $V(P_i)$, 
 the end vertices of $P_i$ are not occupied and alternate vertices in $P_i$ are occupied by guards. 
 \begin{itemize}
   \item Case 1 : $T_1(G, B)=\emptyset$ 
  \begin{enumerate}
   \item[a.] when $n-k$ is odd. \\
   In this case, from Definition~\ref{def:eventful}, it follows  that $l$ is odd. 
  \begin{itemize}
   \item  Suppose $l=1$. Let $P$ be the subpath obtained from $B$ by deleting the edges of $P_1$. It is easy to see that $P$ is an eventful path. 
   If $V(P_1)$ contains only 
 $\lfloor\frac{|V(P_1)|}{2}\rfloor$ guards, then the end vertices of $P_1$ are not occupied by guards and the condition stated in Lemma \ref{lemBigPath} cannot hold for $P$. Therefore, the lemma holds when $l=1$. 
   \item  Suppose $l > 1$.  Then, if $V(P_i)$ and $V(P_{i+1})$ (+ is mod $l$) respectively contains only $\lfloor\frac{|V(P_i)|}{2}\rfloor$ and $\lfloor\frac{|V(P_{i+1})|}{2}\rfloor$ guards and
   the vertex bunch of the path $P$ between the last vertex of $P_i$ and the first vertex of $P_{i+1}$
contains exactly as many guards as mentioned in the first part of 
Lemma~\ref{lemBigPath}, the condition stated in the second part of 
Lemma~\ref{lemBigPath} cannot hold for $P$. 
Since this is true for all $1 \le i \le l$, and the condition stated in Lemma~\ref{lemBigPath} has to simultaneously 
 hold for all subpaths of the cycle that are eventful in $G$, a simple counting argument shows that number of guards in $\mathcal{C}$ should be at least 
 $\left\lceil{\frac{n-k}{2}}\right\rceil+\chi(G, B)$.

Thus, we know that $\evc(G)$ is at least the expression given above. If $T_1(G, B)=\emptyset$, it is easy to show that these many guards are also sufficient. The guards on $V(B)$ can defend any attack on edges of $B$ while keeping $X \cap V(B)$ always occupied. Attacks on edges of $B$-components can be handled, maintaining $x_i$ always occupied and having exactly $\evc_{x_i}(G_i)$ guards on each $B$-component $G_i$. 
\end{itemize}
\item[b.] When $n-k$ is even\\
In this case, $l$ is also even. A similar counting argument will show that $\evc(G)$ is given by the expression in the statement of the lemma. 
\end{enumerate}
\item Case 2 : $T_1(G, B) \ne \emptyset$ 
\begin{enumerate}
 \item[a.] when $n-k$ is even. \\ 
Suppose there are exactly $\frac{n-k}{2}+\chi(G, B)$ guards on $G$. Then, it is not difficult to see that in order to satisfy the condition stated in Lemma~\ref{lemBigPath} 
for all eventful subpaths of the cycle, the configuration of guards should be such that 
\begin{itemize}
\item All cut vertices have guards.
\item while going around the cycle $B$ (discarding the cut vertices), non-cut vertices are alternately guarded and unguarded. 
\item each $G_i \in T_1(G, B)$ contains exactly $\evc_{x_i}(G_i)-1$ guards
\item each $G_i \in T_2(G, B)$ contains exactly $\evc_{x_i}(G_i)$ guards. 
\end{itemize}
Now, consider a Type~1 $B$-component $G_i$. 
By the conditions listed above, $x_i$ must be occupied in every configuration. Hence, there is a sequence of attacks on $G_i$ that would eventually lead to a configuration with at least $\evc_{x_i}(G_i)$ guards on $V(G_i)$ to defend the attack. However, in that configuration, the conditions listed above will not be satisfied. Hence, we need at least one more guard. 
We argue below that with one more guard, we can maintain the following invariants in every configuration:
\begin{itemize}
\item All cut vertices have guards
\item while going around the cycle $B$ (discarding the cut vertices), non-cut vertices are alternately guarded and unguarded
\item one $B$-component $G_i \in T_1(G, B)$ contains exactly $\evc_{x_i}(G_i)$ guards.
\item all other $G_j \in T_1(G, B)$ contain exactly $\evc_{x_j}(G_j)-1$ guards each.
\item each $G_i \in T_2(G, B)$ contains exactly $\evc_{x_i}(G_i)$ guards. 
\end{itemize}
When there is an attack on an edge of a Type~1 $B$-component $G_j$ containing only $\evc_{x_j}(G_j)-1$ guards, we need to reconfigure guards in such a way that after the reconfiguration, the Type~1 $B$-component $G_i$ that presently has $\evc_{x_i}(G_i)$ guards will have one less guard and $G_j$ gets one more guard. Since $G_i$ and $G_j$ are Type~1, this is always possible by an appropriate shifting of guards through the cycle $B$. 
In this way, the invariants stated above can be maintained consistently. Hence, $\evc(G)$ is as given in the statement of the lemma.
\item[b.] when $n-k$ is odd.\\
In this case, note that $\lceil\frac{n-k}{2}\rceil = \lceil\frac{n-k+1}{2}\rceil$. As noted earlier, $l$ is odd and $l \ge 1$.  
Using similar arguments as in the case when there are no Type~1 $B$-components,
we can show that $\evc(G)$ is at least $\left\lceil{\frac{n-k}{2}}\right\rceil+\chi(G, B)$.  With these many guards, it is possible to protect $G$ keeping the following invariants in all configurations.
\begin{itemize}
\item All cut vertices have guards
\item at most one $G_i \in T_1(G, B)$ contains exactly $\evc_{x_i}(G_i)$ guards 
\item all other $G_j \in T_1(G, B)$ contain exactly $\evc_{x_j}(G_j)-1$ guards
\item the maximal uneventful paths $P$ and $Q$ in $B$ that are respectively clockwise and anticlockwise nearest to the Type~1 $B$-component with $\evc_{x_i}(G_i)$ guards (if it exists)
respectively contain $\lfloor\frac{|V(P)|}{2}\rfloor$ and $\lfloor\frac{|V(Q)|}{2}\rfloor$ guards
\item each $G_i \in T_2(G, B)$ contains exactly $\evc_{x_i}(G_i)$ guards. 
\end{itemize}
\end{enumerate}
Hence, in this case also, the lemma holds.
\end{itemize}
 \end{proof}
 
 The following observation gives a method to compute the type of a graph with respect to degree-two vertices in blocks which are cycles.  
 \begin{lemma}\label{lem:type_degree_two}
  Let $B$ be a cycle of $n$ vertices, forming a block of a connected graph $G$. Let $X$ be the set of cut vertices of $G$ and let $k=|X\cap V(B)|$. 
  Suppose each $B$-component in $\mathcal{C}_B(G)$ satisfies the substructure property. 
  Let $v \in V(B) \setminus X$.  The type of $G$ with respect to $v$ can be computed as follows. 
  \begin{itemize}
   \item If $T_1(G, B) \ne \emptyset$ and $n-k$ is even, then $\evc_v(G)=\evc(G)=\evc(G_v^+)$.
   If $T_1(G, B) \ne \emptyset$ and $n-k$ is odd, then $\evc_v(G)=\evc(G)+1=\evc(G_v^+)$.
   In both cases, $G$ is Type~1 with respect to $v$. 
   \item If $T_1(G, B) = \emptyset$ and $n-k$ is even, then $\evc_v(G)=\evc(G)+1=\evc(G_v^+)$
   and $G$ is Type~1 with respect to $v$. If $T_1(G, B) = \emptyset$ and $n-k$ is odd, then $\evc_v(G)=\evc(G)<\evc(G_v^+)=\evc_v(G)+1$ and $G$ is Type~2 with respect to $v$.
  \end{itemize}
\end{lemma}
\begin{proof}
 This can be proved using Lemma~\ref{lem:evc_block} and Lemma~\ref{lem:ContractVertex}.
 Let $G'$ be the graph obtained by deleting the two edges incident on $v$ from $G$ and adding an edge between its neighbors. Let $B'$ be the cycle obtained
by deleting the two edges incident on $v$ from $B$ and adding an edge between its neighbors. 
We have $|X \cap V(B)|=k=|X \cap V(B')|$ and $|V(B')|=|V(B)|-1=n-1$. 
Further, if $\tilde{X}$ is the set of cut vertices of $G_v^+$, then $|\tilde{X}|=|X+1|=k+1$. 
 \begin{itemize}
  \item If $T_1(G, B) \ne \emptyset$, then using Lemma~\ref{lem:evc_block}
 for $G$ and $G'$, we can see that when $n-k$ is odd, $\evc(G')=\evc(G)$ and when $n-k$ is even, $\evc(G')=\evc(G)-1$. 
 Therefore, by Lemma~\ref{lem:ContractVertex}, when $n-k$ is odd, $\evc_v(G)=\evc(G)+1$ and when $n-k$ is even, $\evc_v(G)=\evc(G)$.
 Further, using Lemma~\ref{lem:evc_block} for $G_v^+$, we can see that when $n-k$ is odd, $\evc(G_v^+)=\evc(G)+1$ and 
 when $n-k$ is even, $\evc(G_v^+)=\evc(G)$. Thus, in both cases, $\evc_v(G)=evc(G_v^+)$ and hence, $G$ is Type~1 with respect to $v$.
 \item If $T_1(G, B) = \emptyset$, then using Lemma~\ref{lem:evc_block} for $G$ and $G'$, we can see that
 when $n-k$ is even, $\evc(G')=\evc(G)$ and when $n-k$ is odd, $\evc(G')=\evc(G)-1$. 
 Therefore, by Lemma~\ref{lem:ContractVertex}, when $n-k$ is even, $\evc_v(G)=\evc(G)+1$ and when $n-k$ is odd, $\evc_v(G)=\evc(G)$.
 Further, using Lemma~\ref{lem:evc_block} for $G_v^+$, we can see that when $n-k$ is even, $\evc(G_v^+)=\evc(G)+1$ and when $n-k$ is odd, then $\evc(G_v^+)=\evc(G)+1$.
 Therefore, when $n-k$ is even, $\evc_v(G)=\evc(G_v^+)$ and when $n-k$ is odd, $\evc(G_v^+)=\evc_v(G)+1$.
 Hence, when $n-k$ is even, $G$ is Type~1 with respect to $v$ and when $n-k$ is odd, $G$ is Type~2 with respect to $v$.
 \end{itemize}
\end{proof}
\section{Computing eternal vertex cover number of cactus graphs}\label{sec:cactus}
\begin{theorem}\label{thm:cactus-substructure}
Every cactus graph satisfies substructure property.
\end{theorem}
\begin{proof}
Let $G$ be a cactus graph. 
 The proof is using an induction on the number of cut-vertices in $G$. 
 
 In the base case, $G$ is a  cactus without a cut vertex. Then, $G$ is either a single vertex, a single edge or a simple cycle. 
 In all these cases, the lower bound for the number of guards on $V(G)$, specified by substructure property is equal to the vertex cover number of the respective graphs. Hence,
 the theorem holds in the base case. 
 
 Now, let us assume that the theorem holds for any cactus with at most $k$ cut-vertices. Let $G$ be a cactus with $k+1$ cut vertices, for $k\ge 0$ and let $X$ be the set of cut vertices of $G$.
 Let $v$ be a non-cut vertex of $G$ and $G'$ be an arbitrary $v$-extension of $G$. We need to show that in any eternal vertex cover configuration of $G'$, the number of guards on $V(G)$ is as 
 specified by the substructure property. Since $v$ is a non-cut vertex of the cactus $G$, either it is a degree-one vertex of $G$ or it is a degree-$2$ vertex of $G$ that is in some block $B$ of $G$, where
 $B$ is a cycle. We want to compute a lower bound on the number of guards on $V(G)$ in an arbitrary eternal vertex cover configuration $\mathcal{C'}$ of $G'$.
 
 First, consider the case where $v$ is a degree-one vertex of $G$. Let $w$ be the neighbor of $v$ in $G$ and let $H= G \setminus v$. 
 Since $G$ has at least one cut-vertex, $w$ must be a cut vertex in $G$. 
 \begin{itemize}
  \item When $w$ is a cut-vertex of $H$: Consider any $w$-component $H_i$ of $H$. $H_i$ is a cactus and the number of its cut vertices is less than that of $G$. 
 Hence, $H_i$ satisfies the substructure property. Further, $G'$ is a $w$-extension of $H_i$.    
 In $\mathcal{C'}$, the number of guards on $V(H_i)$ is at least 
 $\evc_w(H_i)-1$ if $H_i \in T_1(H,w)$ and it is $\evc_w(H_i)$ if $H_i \in T_2(H, w)$.  The total number of guards on $V(H)$ is at least 
 $1+\sum_{H_i \in T_1(H, w)}{(\evc_w(H_i)-2)}+ \sum_{H_i \in T_2(H, w)}{(\evc_w(H_i)-1)}$. If $v$ is occupied in $\mathcal{C'}$, then the total number of guards on $V(G)$ is at least
 the same as that of $\evc(G)$ as given by Lemma~\ref{lem:cut-vertex}. If $v$ is not occupied in $\mathcal{C'}$, then to defend an attack on the edge $wv$, a guard from $V(H)$ must move to $v$. 
 Hence, in $\mathcal{C'}$, the number of guards in $V(H)$ should have been one more than the minimum mentioned earlier. Hence, in this case also, the number of guards on $V(G)$ in $\mathcal{C'}$ would have been
 at least $\evc(G)$.

 By Lemma~\ref{lem:cut-vertex}, $\evc(G_v^+)=\evc(G)+1$. 
 By Lemma~\ref{lem:pendent_vertex}, the type of $G$ with respect to $v$ is the same as the type of $H$ with respect to $w$. Hence, if $H$ is Type~1 with respect to $w$, we have  
 $\evc(G_v^+)=\evc_v(G)=\evc(G)+1$. We have seen that the number of guards on $V(G)$ is at least $\evc(G) = \evc_v(G) -1$. 
 If $H$ is Type~2 with respect to $w$, then we have $\evc(G_v^+) > \evc_v(G)$. Since $\evc(G_v^+)=\evc(G)+1$, in this case we must have $\evc_v(G)=\evc(G)$. 
 We have seen that the number of guards on $V(G)$ is at least $\evc(G)=\evc_v(G)$.  Hence, the substructure property holds in both cases.
 \item When $w$ is not a cut vertex of $H$: In this case, $G'$ is a $w$-extension of $H$. 
  By substructure property of $H$, it follows that in configuration $\mathcal{C'}$, the number of guards on $V(G)$ must be at least $\evc_w(H)$ if $H$ is Type~1 with respect to $w$ and at least $\evc_w(H)+1$ if 
  $H$ is Type~2 with respect to $w$.
  By Lemma~\ref{lem:cut-vertex}, when $H$ is Type~1 with respect to $w$, $\evc(G)=\evc_w(H)$ and when $H$ is Type~2 with respect to $w$, $\evc(G)=\evc_w(H)+1$. 
  By Lemma~\ref{lem:pendent_vertex}, the type of $G$ with respect to $v$ is the same as the type of $H$ with respect to $w$.
  Hence, if $H$ is Type~1 with respect to $w$, we have $\evc(G_v^+)=\evc_v(G)=\evc(G)+1$. The number of guards on $V(G)$ is at least $\evc_w(H)=\evc_v(G)-1$.
  If $H$ is Type~2 with respect to $w$, we have $\evc(G_v^+) > \evc_v(G)$. Since $\evc(G_v^+)=\evc(G)+1$, in this case we must have $\evc_v(G)=\evc(G)$. 
  The number of guards on $V(G)$ is at least $\evc_w(H)+1=\evc_v(G)$. Hence, the substructure property holds in both cases.
 \end{itemize}
 Now, consider the case when $v$ is a degree-two vertex of $G$ that is in some block $B$ of $G$, where $B$ is a cycle. Suppose $|V(B)|=n_b$. Let $X$ be the set of cut vertices of $G$ and $k_b=|X \cap V(B)|$. As noted earlier, we have to compute a lower bound on the number of guards on $V(G)$ in an arbitrary eternal vertex cover configuration $\mathcal{C'}$ of $G'$.
 Let $p$ and $q$ respectively be the clockwise and anticlockwise nearest vertices to $v$ in 
 $X \cap V(B)$. Let $P$ be the path in $B$ between $p$ and $q$ that does not contain $v$.
 Note that, every $B$-component of $G$ satisfies substructure property by our induction hypothesis and $G'$ is an extension for each of them. Hence, we have a lower bound on the number of guards on $V(G_i)$, for each $G_i \in \mathcal{C}_B(G)$. Similarly, the condition stated in Lemma~\ref{lemBigPath} needs to be satisfied for each eventful subpath $P'$ of $P$.
 \begin{itemize}
  \item If $T_1(G, B) = \emptyset$: By similar arguments as in the proof of Lemma~\ref{lem:evc_block}, we can see that the number of guards on $V(G)$ in $\mathcal{C'}$ must be at least $\evc(G)$. By Lemma~\ref{lem:type_degree_two}, when $n_b-k_b$ is even, $G$ is Type~1 with respect to $v$ and $\evc_v(G)=\evc(G)+1$. Since the number of guards on $V(G)$ is at least $\evc(G)=\evc_v(G)-1$, we are done. Similarly, when $n_b-k_b$ is odd, $G$ is Type~2 with respect to $v$ and $\evc_v(G)=\evc(G)$ and we are done. 
  \item If $T_1(G, B) \ne \emptyset$ and $n_b-k_b$ is even: By similar arguments as in the proof of Lemma~\ref{lem:evc_block}, we can see that the number of guards on $V(G)$ in $\mathcal{C'}$ must be at least $\evc(G)-1$. By Lemma~\ref{lem:type_degree_two},  $G$ is Type~1 with respect to $v$ and $\evc(G)=\evc_v(G)$. Since the number of guards on $V(G)$ is at least $\evc(G)-1=\evc_v(G)-1$, we are done. 
  \item $T_1(G, B) \ne \emptyset$ and $n_b-k_b$ is odd: By similar arguments as in the proof of Lemma~\ref{lem:evc_block}, we can see that the number of guards on $V(G)$ in $\mathcal{C'}$ must be at least $\evc(G)$. By Lemma~\ref{lem:type_degree_two},  $G$ is Type~1 with respect to $v$ and $\evc_v(G)=\evc(G)+1$.  Since the number of guards on $V(G)$ is at least $\evc(G)=\evc_v(G)-1$, we are done. 
 \end{itemize}
Thus, in all cases, the lower bound on the number of guards on $V(G)$ in an arbitrary eternal vertex cover configuration $\mathcal{C'}$ of $G'$ satisfies the condition stated in substructure property. Hence, $G$ satisfies substructure property.

Thus, by induction, it follows that every cactus satisfies substructure property.
\end{proof}
Now, we have all ingredients for designing a recursive algorithm for the computation of eternal vertex cover number of a cactus, using Lemma~\ref{lem:cut-vertex}. Our algorithm will take a cactus $G$ and a vertex $v$ of $G$ and output $\evc(G)$, $\evc_v(G)$ and the type of $G$ with respect to $v$. If $G$ is a cycle or an edge or a vertex, the answer is trivial and can be computed in linear time. 

In other cases, $G$ has at least one cut vertex. If $v$ is a cut vertex, then we call the algorithm recursively on each $v$-component of $G$ along with vertex $v$. Then, we can use Lemma~\ref{lem:cut-vertex} to compute $\evc(G)$ and $\evc_v(G)$ in constant time from the result of the recursive call. Using the same information from recursive calls, the type of $G$ with respect to $v$ can also be computed using Observation~\ref{obs:type_cut_vertex}. If $v$ is a pendent vertex and $w$ is its neighbor in $G$, then we recursively call the algorithm on $(G\setminus v, w)$. By Lemma~\ref{lem:pendent_vertex}, the type of $G$ with respect to $v$ is the same as the type of $G\setminus v$ with respect to $w$. Moreover, from the proof of Lemma~\ref{lem:pendent_vertex}, we have $\evc_v(G)=\evc_w(G\setminus v)+1$. Further, $\evc(G)=\evc_w(G \setminus v)$, when $G$ is Type~1 with respect to $v$ and 
$\evc(G)=\evc_w(G \setminus v)+1$, when $G$ is Type~2 with respect to $v$. Thus, from the results of the recursive call on $(G\setminus v, w)$, the output can be computed in constant time. In the remaining case, $v$ is a vertex that belongs to a cycle $B$ in $G$. In this case, we recursively call the algorithm for each $B$-component of $G$, along with the respective cut vertices its shares with $B$. Using this information, we can compute $\evc(G)$ using Lemma~\ref{lem:evc_block} in time proportional to the number of $B$-components. We can also compute $\evc_v(G)$ and the type of $G$ with respect to $v$, using  Lemma~\ref{lem:type_degree_two} in time proportional to the number of $B$-components. 

Thus, the algorithm works in all cases and runs in time linear in the size of $G$. Hence, we have the following result.
\begin{theorem}\label{thm:cactus-algorithm}
 Eternal vertex cover number of a cactus $G$ can be computed in time linear in the size of $G$.
\end{theorem}
From the upper bound arguments in the proofs discussed in Section~\ref{sec:basics} and Section~\ref{sec:type-computation}, we can see that with $\evc(G)$ guards determining configurations of guards to keep defending attacks on $G$ is straightforward. 
\section{Extension to other graph classes}\label{sec:extension}
It may be noticed that most of the intermediate results stated in this paper are generic, though we have stated Lemma~\ref{lem:evc_block} and Lemma~\ref{lem:type_degree_two} in a way suitable for handling cactus graphs. 
In this section, we show how to extend the method used for cactus graphs to a graph class which is somewhat more general. We consider connected graphs in which each block is a cycle, an edge or a biconnected chordal graph.

 To generalize the proof of Theorem~\ref{thm:cactus-substructure} for this class, the base case of the proof needs to be modified to handle biconnected chordal graphs as well. The following observation addresses this requirement. 
 \begin{observation}\label{obs:biconnected-chordal-substructure}
  Every biconnected chordal graph satisfies substructure property.
 \end{observation}
 \begin{proof}
  Let $G$ be a biconnected chordal graph and $v \in V(G)$.  
  Let $G'$ be a $v$-extension of $G$ and $C$ be an eternal vertex cover configuration of $G'$.
  If the number of guards on $V(G)$ in $C$ is less than $\mvc_v(G)$, then $v$ is not occupied in $C$. In this configuration, an attack on an edge of $G$ adjacent to $v$ cannot be defended, because when a guard on a vertex of $G$ moves to $v$, some edge of $G$ will be without guards.
  Therefore, in any arbitrary eternal vertex cover configuration $C$ of $G'$, the number of guards on $V(G)$ must be at least $\mvc_v(G)$. 
   
  By a result in~\cite{BCFPRW:CALDAM}, $\evc_v(G) \in \{\mvc_v(G), \mvc_v(G)+1\}$ and $\evc(G_v^+)=\mvc(G_v^+)+1=\mvc_v(G)+1$. If $\evc_v(G)=\mvc_v(G)$ and $\evc(G_v^+)=\mvc_v(G)+1$, then $G$ is Type~2 with respect to $v$ and we need at least $\evc_v(G)=\mvc_v(G)$ guards on $V(G)$. If $\evc_v(G)=\mvc_v(G)+1=\evc(G_v^+)$, then $G$ is Type~1 with respect to $v$ and we need at least $\evc_v(G)-1=\mvc_v(G)$ guards on $V(G)$. Thus, in both cases, the requirements of substructure property are satisfied by $G$.  
 \end{proof}
 The following lemma is a suitable modification of Lemma~\ref{lem:evc_block} and Lemma~\ref{lem:type_degree_two} to handle the new class. 
\begin{lemma}\label{lem:evc_chordal}
Let $B$ be a biconnected chordal graph forming a block of a connected graph $G$ and let $X$ be the set of cut vertices of $G$. 
 Suppose each $B$-component $G_i$ of $G$ that belongs to $\mathcal{C}_B(G)$ satisfies the substructure property.
 If $T_1(G, B)=\emptyset$, 
 then $\evc(G)=\evc_{ X\cap V(B)}(B) +\chi(G, B) -|X \cap V(B)| $ and
 $\evc(G)=\mvc_{ X\cap V(B)}(B) + 1 +\chi(G, B) -|X \cap V(B)|$ otherwise.
 Further, if $\chi(G, B)$ is known, then for any $v \in V(B)$, the type of $G$ with respect to $v$ can be computed in time quadratic in the size of $B$. 
 \end{lemma}
 \begin{proof}
  Since $B$ a biconnected chordal graph, if for every $v \in V(B) \setminus X$,  $\mvc_{ (X\cap V(B))\cup \{v\}}(B) = \mvc_{ X\cap V(B)}(B)$, then $\evc_{ X\cap V(B)}(B)= \mvc_{ X\cap V(B)}(B)$ and $\evc_{ X\cap V(B)}(B)= \mvc_{ X\cap V(B)}(B)+1$ otherwise~\cite{BCFPRW:CALDAM}. 
  Consider any eternal vertex cover class $\mathcal{C}$ of $G$. By substructure property of $B$-components of $G$, it follows that in any configuration of $\mathcal{C}$, the number of guards on $\bigcup_{G_i \in \mathcal{C}_B}{V(G_i)}$ is at least $\chi(G, B)$. To cover edges of the induced subgraph of $G$ on $V(B) \setminus X$, the number of guards required is at least $\mvc(B \setminus X)= \mvc_{ X\cap V(B)}(B)-|X \cap V(B)|$. Hence, the total number of guards on $V(G)$ must be at least $\mvc_{ X\cap V(B)}(B)-|X \cap V(B)|+\chi(G, B)$.  Further, if there is a vertex $v \in V(B) \setminus X$ for which $\mvc_{ (X\cap V(B))\cup \{v\}}(B) \ne \mvc_{ X\cap V(B)}(B)$, then in any configuration of $\mathcal{C}$ in which $v$ is occupied, the total number of guards on $V(G)$ must be at least $\mvc_{ X\cap V(B)}(B)+1-|X \cap V(B)|+\chi(G, B)$. Hence, in all cases, $\evc(G) \ge \evc_{ X\cap V(B)}(B) +\chi(G, B) -|X \cap V(B)|$.
    \begin{itemize}
   \item[a.]When $T_1(G, B)=\emptyset$: In this case, it is easy to show that $\evc(G) \le \evc_{ X\cap V(B)}(B) +\chi(G, B) -|X \cap V(B)|$.
   \item[b.]When $T_1(G, B) \ne \emptyset$: In this case, by repeated attacks on edges of $V(G_i)$ for some $G_i \in T_1(G, B)$, eventually a configuration which requires $\chi(G, B)+1$ guards on $\bigcup_{G_i \in \mathcal{C}_B}{V(G_i)}$ can be forced. Hence, $\evc(G) \ge \mvc_{ X\cap V(B)}(B)+1+\chi(G, B) -|X \cap V(B)|$. Since $\evc_{ X\cap V(B)}(B) \le \mvc_{ X\cap V(B)}(B)+1$, it is not difficult to also show that $\evc(G) \le \mvc_{ X\cap V(B)}(B)+1+\chi(G, B) -|X \cap V(B)|$. 
  \end{itemize}
  In both cases, the value of $\evc(G)$ is as stated in the lemma.  
  For deciding the type of $G$ with respect to a vertex $v \in V(B)$, we need to compare $\evc_v(G)$ and $\evc(G_v^+)$. For computing $\evc(G_v^+)$, we can use the formula given by the first part of the lemma for the graph $\evc(G_v^+)$. Using similar arguments as in the proof of the first part of the lemma, we get the following: if $T_1(G, B)=\emptyset$, then $\evc_v(G)=\evc_{(X\cap V(B))\cup \{v\}}(B) +\chi(G, B) -|X \cap V(B)| $ and
 $\evc_v(G)=\mvc_{ (X\cap V(B))\cup \{v\}}(B) + 1 +\chi(G, B) -|X \cap V(B)|$ otherwise. By computing $\mvc_{ (X\cap V(B))}(B)$ and $\mvc_{ (X\cap V(B))\cup \{v\}}(B)$ for each $v \in V(B)$, $\evc_v(G)$ and $\evc(G_v^+)$ can be computed. Since minimum vertex cover of a chordal graph can be computed in linear time, the total time required for computing $\evc_v(G)$ and $\evc(G_v^+)$ this way is possible in time quadratic in the size of $B$.  
 From the values of $\evc_v(G)$ and $\evc(G_v^+)$, the type of $G$ with respect to $v$ can be inferred. 
\end{proof}

 Now, similar arguments as in the proof of Theorem~\ref{thm:cactus-substructure} and Theorem~\ref{thm:cactus-algorithm} yields:
 \begin{theorem}
   Suppose $G$ is a connected graph in which each block is a cycle, an edge or a biconnected chordal graph. Then, $G$ satisfies substructure property and the eternal vertex cover number of $G$ can be computed in time quadratic in the size of $G$.
 \end{theorem}
 \section{A cactus for which other lower bounds are weak}\label{sec:example}
 \begin{figure}[h]
 \begin{center}
 \includegraphics[scale=0.7]{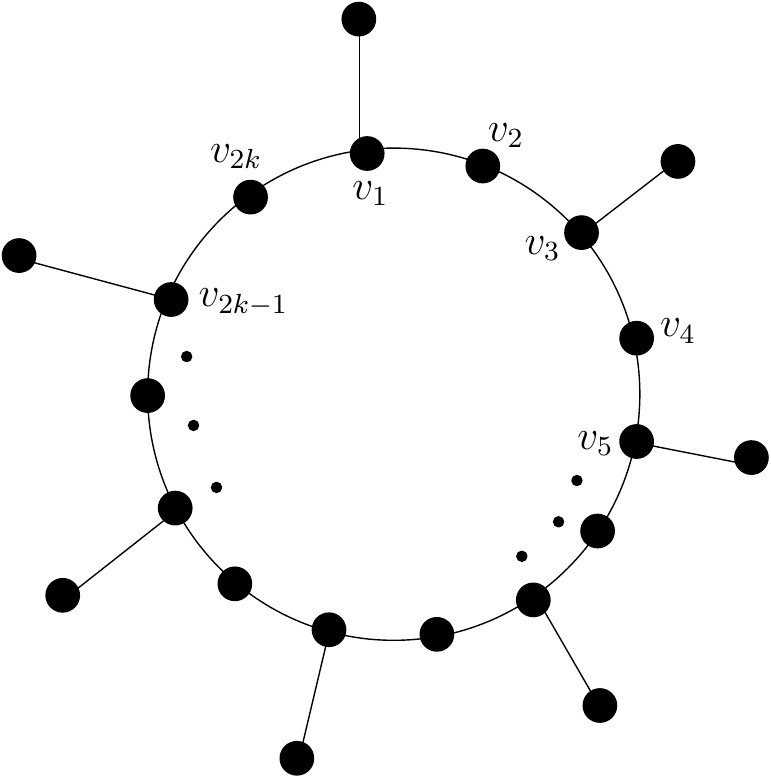}
 \end{center}
\caption{A cactus $G$ in which pendent vertices are attached to alternate vertices of an even cycle on $2k$ vertices. The cactus $G$ has a minimum vertex cover of size $k$ containing all cut vertices of $G$ and $\evc(G)=k+\lceil\frac{k+1}{2}\rceil$. }
\label{fig:cactus_example}
 \end{figure}
It is known that if $G$ is a graph with $X$ being the set of its cut vertices, and $\mvc_X(G)$ is the minimum cardinality of a vertex cover of $G$ that contains all vertices of $X$, then $\evc(G) \ge \mvc_X(G)$ \cite{evc_lb_Arxiv}. For cycles, 
this lower bound is equal to the eternal vertex cover number and for trees the difference of the eternal vertex cover number and this lower bound is at most one.   
But, it is interesting to note that the lower bound could be much smaller than the optimum for some cactus graphs. An example of this is shown in Figure~\ref{fig:cactus_example}. Using the formula given by Lemma~\ref{lem:evc_block}, we can show that the eternal vertex cover number of this graph $G$ is $k+\lceil\frac{k+1}{2}\rceil$. However, for this graph, $\mvc_X(G)$ is only $k$ making $\evc(G) > 1.5 \mvc_X(G)$.  
\section{Discussion and Open problems}
The lower bounding method based on the substructure
property of cactus graphs presented here shows that the method could 
be potentially useful in 
obtaining efficient algorithms for computing the eternal vertex cover
number of graphs that belong to classes 
for which other known lower bound techniques 
are not effective.  Though the substructure property has been applied here in context of cactus graphs, we believe
that this property or a generalization of it holds for fairly large classes of graphs 
(if not all graphs). We list below a few questions that appear interesting 
in this context.
\begin{itemize}
\item Do all graphs satisfy the substructure property? If not, can we characterize graphs that satisfy this property?
\item Is it possible to generalize substructure property and use it for the computation of eternal vertex cover number of outerplanar graphs and bounded treewidth graphs?
\item In the restricted version of the eternal vertex cover problem in which at most one guard is permitted on a vertex in any configuration, is it true in general that for any graph $G$ and vertex $x$ of $G$, $\evc(G) \le evc(G_x)^+$?
\end{itemize}
The last question looks deceptively simple (and is trivial if multiple guards are allowed on a vertex in a configuration).   For graphs that satisfy substructure property, the answer to this question is affirmative (by Remark~\ref{rmk:type_noncut} and Observation~\ref{obs:type_cut_vertex}). But, we do not know the answer in general.
\section{Acknowledgments}
We would like to thank Ilan Newman, University of Haifa, for introducing the problem and for his useful suggestions.

\end{document}